\documentclass[a4wide]{article}

\usepackage[english]{babel} 
\usepackage[utf8]{inputenc}  
\usepackage[T1]{fontenc}       

\usepackage{amssymb,amsfonts,amsthm}

\usepackage[hscale=0.75,vscale=0.8]{geometry}

\usepackage{enumerate}
\usepackage[all]{xy}
\usepackage[hyphens]{url}
\usepackage{commath}
\usepackage{authblk}

\newcommand{\Q}{{\mathbb Q}}
\newcommand{\Z}{{\mathbb Z}}

\newcommand{\R}{{\mathbb R}}
\newcommand{\OO}{{\mathcal O}}

\DeclareMathOperator{\Disc}{Disc}

\DeclareMathOperator{\Norm}{N}
\DeclareMathOperator{\val}{val}

\DeclareMathOperator{\cont}{cont}

\theoremstyle{plain}
\newtheorem{theo}{Theorem}[section]
\newtheorem{theoremletter}{Theorem}[section]

\newtheorem{lem}[theo]{Lemma}
\newtheorem{prop}[theo]{Proposition}

\theoremstyle{definition}
\newtheorem{deff}[theo]{Definition}
\newtheorem{ex}[theo]{Example}

\theoremstyle{remark}
\newtheorem{rem}[theo]{Remark}

\makeatletter
\newcommand{\pushright}[1]{\ifmeasuring@#1\else\omit\hfill$\displaystyle#1$\fi\ignorespaces}
\newcommand{\pushleft}[1]{\ifmeasuring@#1\else\omit$\displaystyle#1$\hfill\fi\ignorespaces}
\makeatother

\title{Some mathematical remarks \\on the polynomial selection in NFS 
}
\author[1,2]{Razvan Barbulescu \thanks{razvan.barbulescu@inria.fr}}
\author[1,3]{Armand Lachand\thanks{armand.lachand@univ-lorraine.fr}}
\affil[1]{Universit{\'e} de Lorraine}
\affil[2]{CNRS, INRIA}
\affil[3]{Institut Elie Cartan Lorraine (CNRS/Univ. Lorraine)}
\date{}

\begin{document}
\maketitle

\abstract{In this work, we consider the proportion of smooth (free of large prime factors) values of a binary form $F(X_1,X_2)\in\Z[X_1,X_2]$. In a particular case, we give an asymptotic equivalent for this proportion which depends on $F$. This is related to Murphy's $\alpha$ function, which is known in the cryptographic community, but which has not been studied before from a mathematical point of view. Our result proves that, when $\alpha(F)$ is small, $F$ has a high proportion of smooth values. This has consequences on the first step, called polynomial selection, of the Number Field Sieve, the fastest algorithm of integer factorization.}
\section{Introduction}

Smooth -- or friable --numbers, defined as integers whose prime factors are smaller than a given bound, are a celebrated topic in analytic number theory and have a key importance in cryptography today. In this work we are motivated by the Number Field Sieve (NFS), the fastest algorithm of integer factorization~\cite{NFS}. 

Briefly, if $N$ is an integer to be factored, NFS can be summarized as follows. In the first step, called polynomial selection, we select two irreducible polynomials with integer coefficients $f$
and $g$, which have a common root $m$ modulo $N$, i.e. $f(m)\equiv0\equiv
g(m)\pmod N$. In
the next step, we fix a parameter $B$ and we search for $B$ pairs of coprime
integers $(a,b)$ such that $F(a,b):=b^{\deg f}f(a/b)$ and $G(a,b):=b^{\deg
g}g(a/b)$ are $B$-smooth --  an integer $n$ is $B$-smooth if its greatest prime factor, denoted by $P(n)$, satisfies $P(n)\leq B$. The collected pairs allow us to obtain a $B\times B$
linear system over $\Z/2\Z$. Next, we compute a linear combination of the rows of the system. By a square root computation in a number field, we find a non-trivial solution of the equation $x^2\equiv
y^2\bmod N$, which gives a non-trivial factor of $N$.

Computing the complexity of the algorithm requires to find the distribution of coprime pairs $(a,b)$ which are smooth with respect to two binary forms $F$ and $G$, i.e. $F(a,b)$ and $G(a,b)$ are smooth for two irreducible homogeneous polynomials $F$ and $G$ with integer coefficients. In the sequel, small caps letters $f$ and $g$ denote polynomials and capital letters denote the associated binary forms.

 The distribution of $B$-smooth integers has made the object of abundant works (for an overview, we refer to \cite{HT93} and \cite{Gr08}). For example, Hildebrand proved in \cite{Hi86} an asymptotic formula in the region
\begin{equation}\label{Hepsilon}
x\geq 3
,\qquad \exp\left((\log\log x)^{5/3+\varepsilon}\right)\leq B\leq x\tag{$H_{\varepsilon}$}.
\end{equation}
\begin{theoremletter}\label{Theoreme Hildebrand}For any fixed $\varepsilon>0$ and uniformly for $(x,B)$  in the region (\ref{Hepsilon}),  we have 
\begin{equation*}\label{estimation Hildebrand}
\Psi(x,B):=\#\left\{n\in[1,x]: P(n)\leq B 
\right\}=x\rho(u)\left(1+O\left(\frac{\log(u+1)}{\log B}\right)\right),
\end{equation*} 
where  $u:=\frac{\log x}{\log B}$ and  $\rho$ denotes the Dickman function, namely the one defined by the delay differential equation
\begin{align*}
\left\{\begin{array}{ll}
u\rho'(u)+\rho(u-1)=0&\text{ if }u>1,\\
\rho(u)=1&\text{ if }0\leq u \leq 1.
\end{array}
\right.\end{align*}
\end{theoremletter}
A few years later, Saias refined this result by giving an asymptotic expansion of $\Psi(x,B)$.
\begin{theoremletter}[Main corollary,\cite{Sa89}]\label{developpement asymptotique Saias}
There exists $C>0$ such that, for any fixed $J\geq0$, $\varepsilon>0$ and  
 uniformly for $(x,B)$ in the region (\ref{Hepsilon}) and such that
 \begin{align*}
  0<u<J+1\Rightarrow (u-\lfloor u\rfloor) > C(J+1)\frac{\log\log B}{\log B},
\end{align*} we have
  \begin{equation*}\label{estimation Saias}
 \Psi(x,B) =
  x\left(\sum_{j=0}^J\gamma_j\frac{\rho^{(j)}(u)}{(\log B)^j}+
O\left(\rho(u)\left(\frac{\log(u+1)}{\log B}\right)^{J+1}
\right)\right),
\end{equation*}
where  $\gamma_j$ are the coefficients of the Taylor series in $s=0$ of $\frac{s\zeta(s+1)}{s+1}$.
In particular, we have  $
\gamma_1=
\gamma-1.$
\end{theoremletter}
 Let $x$ and $B$ be two given integers, $F(X_1,X_2)\in\Z[X_1,X_2]$ a binary form and $\mathcal{K}$ a compact subset of $\R^2$ whose boundary is a continuous closed curve with piecewise continuous derivatives. By $x\mathcal{K}$ we denote the set $\mathcal{K}$ rescaled by a factor $x$. In order to study the distribution of the $B$-smooth integers of the form $F(a,b)$ for coprime integers $a$ and $b$,  we consider the cardinal  $\Psi^{(1)}_{F}(\mathcal{K},x,B)$ defined by
\begin{align*}
\Psi^{(1)}_{F}(\mathcal{K},x,B):=\#\left\{(a,b)\in x\mathcal{K}:\gcd (a,b)=1\text{ and } P(F(a,b))\leq B\right\}.
\end{align*} In \cite{BBDT12}, Balog, Blomer, Dartyge and Tenenbaum developed an argument which can be easily adapted to show the following result. \begin{theoremletter}\label{theoreme BBDT}
Let $\mathcal{K}$ be a compact subset of $\R^2$ whose boundary is a continuous closed curve with piecewise continuous derivatives, $k\geq1$ and $F_1(X_1,X_2),\dots,F_k(X_1,X_2)\in\mathbb{Z}[X_1,X_2]$ 
 some integral and irreducible binary forms of degree  $d_1\geq\dots\geq d_k$.  There exists $u(d_1,\dots,d_k)$ in the interval $\left(1/d_1,+\infty\right)$ with the following property. For any fixed $u<u(d_1,\dots,d_k)$, there exists a constant $c_{F_1,\dots,F_k,\mathcal{K}}(u)$ such that,    
 for $B\geq x^{1/u}\geq2$, we have
\begin{equation*} \label{minoration friable}
 \Psi^{(1)}_{F_1\dots F_k}(\mathcal{K},x,B)\geq
c_{F_1,\dots,F_k,\mathcal{K}}(u) x^2.
\end{equation*}
More precisely, one can take
\begin{equation*}\label{amelioration constante}
 u(d_1,\dots, d_k):=\left\{\begin{array}{ll}
+\infty&\text{if }k\geq2\text{ and } d_1+\dots+d_k\leq3,\\
e^{\frac{1}{2}}&\text{if }k=1\text{ and } d_1=3.
                  \end{array}
\right.
\end{equation*}
\end{theoremletter}

 

It is common to make the assumption that  integers represented by a given binary form  
have
the same probability to be $B$-smooth as arbitrary integers of the same size. Consequently, in the light of Theorem \ref{Theoreme Hildebrand}, we conjecture that, in a domain to be made precise, we have
\begin{equation}\label{conjecture friabilite}
\Psi^{(1)}_{F_1\dots F_k}(\mathcal{K},x,B)\sim \frac{6}{\pi^2}\mathcal{A}(\mathcal{K})x^2\rho(d_1u)\dots\rho(d_ku),
\end{equation}
where $\mathcal{A}(\mathcal{K})$ denotes the area of $\mathcal{K}$.
A similar formula was proven by the second author (\cite{La1+2} and \cite{La3})   when $d_1+\dots +d_k\leq 3$.

Note that the right hand member of Equation~\ref{conjecture friabilite} does not depend on the binary forms $F_1$, $\ldots$, $F_k$. In the current state of research, it seems out of reach to obtain in the general case an equation in which both members depend on the binary forms. In Theorem \ref{theoreme principal quadratique} we refine Theorem~\ref{theoreme BBDT} in the case $k=1$ and $d_1=2$ by making explicit the first approximation term. Since this term depends on the polynomial $f$, it can be used in the polynomial selection stage of NFS, which is done as follows. Using one of the two methods of Kleinjung~(\cite{Kle06},\cite[Sections 4.1]{Bai11} and \cite{Kle08},\cite[Section 4.2]{Bai11}), one generates a large number of pairs of polynomials $f$ and $g$, such that $f$ is irreducible and $g$ linear. For each pair of polynomials, one computes Murphy's $\mathbb{E}(F,G)$ or Murphy's $\alpha(f)$ for the associated binary forms, as defined in~\cite{Mur99}.  Hence one can make a model of the polynomial selection as a random trial of polynomials from a set 
 \begin{equation}\label{definition EDI}E(d,\textbf{I})=\left\{f=\sum_{i=0}^df_iX^i\in\Z[X]\mid f\text{ is irreducible},\forall i, f_i\in I_i\right\},
\end{equation}
 where $\textbf{I}=\prod_{i=0}^dI_i$ is a $(d+1)$-tuple of intervals.
 
Murphy's $\alpha$ is the main object in this article. It is hard to determine when it was proposed in the cryptographic community, but it was known to Montgomery in 1996~\cite{Boe96}. In his thesis, Murphy~\cite{Mur99} introduced $\alpha(f)$ as the sum of a series and gave evidence that, when $\alpha(f)$ is small, $F$ has a high proportion of smooth values. It is computed using the number of roots of $f$ modulo each prime power $p^k$. Based on $\alpha(f)$, one can compute Murphy's $\mathbb{E}(F,G)$, which takes into account the real roots of $f$ and $g$, but it is more costly to compute and not much more accurate than $\alpha(f)$. Also note that, $\alpha$ does not depend on the linear polynomial $g$ since, based on experiments, one can make the conjecture that $g$ has a small influence on the formula of Equation~\eqref{conjecture friabilite}. A thorough development on the polynomial selection from a cryptographic perspective is due to
Bai~\cite{Bai11}. 

\paragraph{Outline}

In Section~\ref{sec:alpha}, we give a rigorous definition of $\alpha(f)$. The mean value of $\alpha(f)$ over $E(d,\textbf{I})$ will be the main goal of Section~\ref{sec:average}. In the last section, we introduce a modification
of NFS. It allows us to obtain a rigorous result on the proportion of smooth elements in number fields of arbitrary degree and then to show that $\alpha(f)$ effectively occurs in the proportion of smooth values of a binary form of degree~$2$. 

\paragraph{Notation} 

In what follows, $K$ stands for a number field and $d_K$, 
 $\mathcal{O}_K$, $U_K$, $G_K$, $\zeta_K$ and $\lambda_K$ denote respectively its degree, 
  ring of integer, unit group, class group, Dedekind zeta function and  residue of $\zeta_K$. The letters $p$, $\mathfrak{p}$ and $\mathfrak{I}$ denote respectively a rational prime, a prime ideal and an arbitrary ideal of $\mathcal{O}_K$.

\section{Definition and convergence of Murphy's $\alpha(f)$}\label{sec:alpha}
From a cryptographic point of view, Theorem~\ref{theoreme principal quadratique}, proved in Section~\ref{The smoothness probability : imaginary quadratic case}, states that $\alpha(f)$ is a good indicator of a polynomial's efficiency for NFS when $f$ is quadratic. In this section we show that it has two properties which are equally important: it has an easy-to-compute formula and it is defined by a series with a high speed of convergence.

\subsection{Definition of $\alpha(f)$}

Murphy introduced $\alpha$ explicitly for arbitrary polynomials, but he gives credit to Montgomery for using the formula in the case of quadratic polynomials~\cite{Boe96}. One can find the formula of $\alpha$ by the following heuristic argument. For any integer $n$ and bound $C$, the $C$-sifted part of $n$ is the largest divisor of $n$ without prime factors less than $C$. For a bound $B$, the $B$-smooth part of $n$ is the largest $B$-smooth divisor of $n$. Experiments show that one can obtain a good guess of $\Psi^{(1)}_{F}(\mathcal{K},x,B)$ by the following empirical method:
\begin{enumerate}
\item Choose a large constant $C$ and compute the average value $\text{cont}(F,C)$ of the logarithm of the $C$-smooth part of the values of $F$. Define $\alpha(F,C)$ as the average value of the logarithm of the $C$-smooth part of a random integer minus $\text{cont}(F,C)$.
\item Approximate $\Psi^{(1)}_{F}(\mathcal{K},x,B)$ by the cardinality of $x\mathcal{K}$ times the probability of a random $C$-sifted integer of size $\left(\max_{(a,b)\in x\mathcal{K}}|F(a,b)|+\alpha(f,C)\right)$ to be $B$-smooth.
\end{enumerate}
This suggests to define $\alpha$ as in the definition below.  In the sequel, $f$ is
a polynomial in $\Z[X]$ such that $\Disc(f)\neq 0$ and $p$ is a
prime. The associated binary form $F$ is defined by $F(X_1,X_2)=X_2^{\deg(f)}f(X_1/X_2)$. 
\begin{deff} For any prime $p$ we define, if it exists,
\begin{equation}
\alpha_p(f)=(\log p)\left(\frac1{p-1}-\text{cont}_p(f)\right), \label{definition alpha p f}
\end{equation}
with
\begin{equation*}
\text{cont}_p(f)=
\lim_{x\rightarrow\infty}\frac{\sum_{(a,b)\in[1,x]^2,\gcd(a,b,p)=1}\val_p F(a,b)}{\#\left\{(a,b)\in
[1,x]^2: \gcd(a,b,p)=1\right\}}.
\end{equation*}

Under the reserve of proving the convergence of the series below, we define
\begin{equation*}
\alpha(f)=\sum_{p\text{ prime}}\alpha_p(f).
\end{equation*}
\end{deff}

To get an other expression for $\text{cont}_p(f)$, we can split the region 
\begin{align*}
\left\{(a,b)\in
[1,x]^2: \gcd(a,b,p)=1\right\}
\end{align*}
in congruence classes modulo $p^k$ and try to approximate
\begin{align*}
\#\left\{(a,b)\in
[1,x]^2: \gcd(a,b,p)=1, p^k|F(a,b)\right\}
\end{align*}
by 
\begin{align*}
\frac{x^2}{p^{2k}}\#\left\{(a,b)\in
[1,p^k]^2: \gcd(a,b,p)=1, p^k|F(a,b)\right\}
\end{align*}
This procedure is essentially the object of Lemma \ref{prop:alpha_p}.
Before doing this, we can remark that 
\begin{equation}\label{formule gamma F}
\#\left\{(a,b)\in
[1,p^k]^2: \gcd(a,b,p)=1, p^k|F(a,b)\right\}=
\varphi(p^k)n_{p^k}(f),
\end{equation}
where \begin{eqnarray*}
		n_{p^k}(f)&=&\#\left\{ a\in [0,p^k-1]:
	f(a)\equiv 0\mod p^k\right\}\\ 
	&+&\#\left\{b\in [0,p^k-1]:b\equiv0\mod p,~F(1,b)\equiv
	0\mod p^k\right\}.
\end{eqnarray*}

Nagell~\cite{Na21} proved what survives of Hensel's lemma when the hypothesis on the derivative fails. We adapt his result to obtain an upper bound of $n_{p^k}$ in a similar way one would in the case when Hensel's lemma applies. 
\begin{lem}\label{Nagel}
If $p$ does not divide $\Disc(f)$
, then $n_{p^k}(f)=n_p(f)$.
In the general case, for any prime $p$ and $k\geq1$, we have
\begin{align*}
n_{p^k}(f)\leq  2\deg(f)p^{\min(2\val_p(\Disc(f)),k)}.
\end{align*}
\end{lem}
\begin{proof}
The first assertion is a direct consequence of [\cite{Na21},Theorem 1] which asserts that
\begin{align*}
\#\left\{ a\in [0,p^k-1]:
	f(a)\equiv 0\mod p^k\right\}= \#\left\{ a\in [0,p-1]:
	f(a)\equiv 0\mod p\right\}.
	\end{align*}
	
	In the proof of   [\cite{Na21},Theorem 2], it is shown that
\begin{align*}
\#\left\{ a\in [0,p^k-1]:
	f(a)\equiv 0\mod p^k\right\}\leq 
	\deg (f)p^{\min(2\val_p(\Disc(f)),k)}.
	\end{align*}	
	When applied to $f(x)=F(x,1)$ and $\overline{f}=F(1,x)$, this implies the second assertion.
	\end{proof}

\begin{prop}\label{prop:alpha_p}
 We have, for every prime $p$,
 \begin{align*}
\alpha_p(f)=\log	p\left(\frac{1}{p-1}-\frac{p}{p+1}\sum_{k\geq
				1}\frac{n_{p^k}(f)}{p^k}\right).
\end{align*}
\end{prop}

\begin{proof}

We first focus on the numerator of $\text{cont}_p(f)$. Let $x$ be a sufficiently large integer. One can choose $k_0$ such that 
$x^{2/3} \leq p^{k_0}\leq px^{2/3}$.
 We  write 
 \begin{align*}
  \sum_{\substack{1\leq a,b\leq x\\(a,b)=1}}\val_p(F(a,b))&=
  \Sigma_1(p,x)+\Sigma_2(p,x)
 \end{align*}
 with
 \begin{align*}
  \Sigma_1(p,x)=\sum_{k\leq k_0}\sum_{1\leq a,b\leq x}\#\left\{
  1\leq a,b\leq x:(a,b,p)=1\text{ and }p^k|F(a,b)\right\}
 \end{align*}
 and
  \begin{align*}
  \Sigma_2(p,x)=\sum_{ k\geq k_0+1}\sum_{1\leq a,b\leq x}\#\left\{
  1\leq a,b\leq x:(a,b,p)=1\text{ and }p^k|F(a,b)\right\}.
 \end{align*}
 
In view of the formula (\ref{formule gamma F}), we can use Lemma 3.2 of \cite{Da99} to deduce that
\begin{align*}
	\Sigma_1(p,x)&=
 \sum_{k\leq k_0} \frac{\#\left\{(a,b)\in
[1,p^k]^2: \gcd(a,b,p)=1, p^k|F(a,b)\right\}}{p^{2k}}x^2\qquad\qquad\qquad\\&\pushright{+O\left( xp^{k_0/2}(k_0\log
 p)^{\nu_f}+p^{k_0}(k_0 \log p)^{2\deg f}\right)}\\
 &=\left(1-\frac{1}{p}\right)
 \sum_{k\leq k_0} \frac{n_{p^k}(f)}{p^{k}}x^2+O\left( x^{4/3}(\log x)^{\nu_f}\right),
\end{align*}
where $\nu_f=\deg f\left( 1+2\deg f\right)^{\deg f+1}$.

On the other hand, 
since $\val_p\left(F(a,b)\right)\ll\log x$, we can use Lemma \ref{Nagel} and again Lemma 3.2 of \cite{Da99} to deduce that
\begin{align*}
 \Sigma_2(p,x)&\ll\log x\sum_{1\leq a,b\leq x}\#\left\{
  1\leq a,b\leq x:(a,b,p)=1\text{ and }p^{k_0}|F(a,b)\right\}\\
  &\ll(\log x)x^2\frac{\#\left\{(a,b)\in
[1,p^{k_0}]^2: \gcd(a,b,p)=1, p^{k_0}|F(a,b)\right\}}{p^{2k_0}}+O\left( x^{4/3}(\log x)^{\nu_f}\right)\\
 &\ll x^{4/3}(\log x)^{\nu_f} .
\end{align*}

Finally, we note that
\begin{align*}
 \#\left\{1\leq a,b\leq x:(a,b,p)=1\right\}&=
\#\left\{1\leq a,b\leq x:p\nmid a\right\}+
 \#\left\{1\leq a,b\leq x:p|a\text{ and }p\nmid b\right\}\\
 &= \left(1-\frac{1}{p^2}\right)x^2+O(x).\end{align*}
The result follows when $x$ tends to infinity since then $k_0$ tends to infinity.
\end{proof}

\subsection{Convergence of $\alpha(f)$}

The formula of $\alpha_p(f)$ gets a simple form when $p$ does not divide $\Disc(f)$ nor the leading coefficient of $f$. 
Indeed, Lemma~\ref{Nagel} and Proposition~\ref{prop:alpha_p} imply that, for such primes $p$, we have
\begin{equation}\label{regular prime}
\alpha_p(f)=
\log p\left(\frac{1}{p-1}-\frac{n_p(f)}{p-1}\left(\frac{p}{p+1}\right)\right).
\end{equation}
Let $\omega$ be a root of $f$, $K$ the rupture field of $f$ and $\tilde{\omega}:=F(1,0)\omega$ an integer of $K$. It follows from a result of Dedekind \cite{De78} that, for any prime $p$ which not divide $F(1,0)$ nor the index 
$[\mathcal{O}_K:\Z[\tilde{\omega}]]$, $n_p(f)$ is the number of ideals $\mathfrak{p}$ such that $N(\mathfrak{p})=p$. This suggests to put \begin{displaymath}
p_0=\max\left\{p\text{ prime}: p|F(1,0)F(0,1)\text{ or }p|\Disc(F)\text{ or }p|[\mathcal{O}_K,\Z[\tilde{\omega}]]\right\}.
\end{displaymath}
After the previous discussion, the problem of convergence of $\alpha(f)$ is reduced to showing the convergence
of the series
\begin{align*} 
\sum_{p}\log p\left(\frac{1}{p-1}-\frac{n_p(K)}{p-1}\left(\frac{p}{p+1}\right)\right)
\end{align*}
where $n_p(K)$ denotes the number of ideals $\mathfrak{p}$ such that $N(\mathfrak{p})=p$.

We first remark that, for any $X\geq2$, we can write
\begin{align*}
 \sum_{p\leq X}\log p\left(\frac{1}{p-1}-\frac{n_p(K)}{p-1}\left(\frac{p}{p+1}\right)\right)
&=\sum_{p\leq X}
\frac{\log p}{p}\left(1-n_p(K)\right)
+\sum_{p\leq X}\frac{\log p}{p(p-1)}\left(1-\frac{n_p(K)}{p+1}\right).
\end{align*}

On the one hand, 
 from the trivial estimation $|n_p(K)|\leq n_K$ and the Chebyshev estimation
 \begin{equation}\label{Chebychev estimate}
 \sum_{p\leq X}\log p\leq eX
 \end{equation}
 with $e=1.01624 $ (see Theorem 9 of \cite{RS62}),
 we can
 use a summation by parts to get, for any 
$X_2\geq X_1\geq n_K$, 
\begin{align*}\label{serie abs. conv.}
 \nonumber\sum_{X_1<p\leq X_2}\frac{\log p}{p(p-1)}\left|1-\frac{n_p(K)}{p+1}\right|
&\leq \sum_{X_1<p\leq X_2}\frac{\log p}{p(p-1)}\\
&\leq 
\frac{3e}{X_1-1}. \end{align*} 
On the other hand, we can write, again with a summation by parts,
\begin{align*}
 \sum_{X_1<p\leq X_2}
\frac{\log p}{p}\left(1-n_p(K)\right) 
&= \frac{R(X_2)}{X_2}- \frac{R(X_1)}{X_1}+ \int_{X_1}^{X_2}
\frac{R(t)}{t^2}\text{d}t,
\end{align*}
where $R$ is the rest term defined by  \begin{align*}
      R(t):=\sum_{p\leq t}(1-n_p(K))\log p.
     \end{align*}

Therefore, it suffices to use a sufficiently sharp estimation of $R(t)$, which is the object of the next theorem. On the
one hand, we can obtain a very sharp estimation using the Riemann
hypothesis for $\zeta_K$ and $\zeta_{\mathbb{Q}}$. But on the other hand, we have a good estimation
relying on no assumptions. 
\begin{theo}[Theorem 9.2 of \cite{LO77}]\label{effectif}
	\begin{enumerate}
		\item There exists an absolute effectively computable
			constant $c_1>0$ such that, if
			$X\geq\exp\left(4d_K(\log \Disc(K))^2\right)$, then
\begin{align*}
\left|\sum_{N(\mathfrak{p}^k)\leq X}\log N(\mathfrak{p})-X+ \frac{X^{\beta(K)}}{\beta(K)}\right| 
\leq X\exp\left(-c_1d_K^{-1/2}(\log X)^{1/2}\right),
\end{align*}
where $\beta(K)$ denotes the largest real zero of $\zeta_K$ in the interval
$(0,1)$ if it exists and $1/2$ otherwise.

\item Moreover, if the Riemann Hypothesis holds for $\zeta_K$, there exist explicit constants $a_K$, $b_K$ and $c_K$ such that, for $X\geq2$, we have
\begin{align*}
\left|\sum_{N(\mathfrak{p}^k)\leq X}\log N(\mathfrak{p})-X\right| 
\leq 
X^{1/2}\left(a_K+b_K\log X+c_K(\log X)^2\right).
\end{align*}
\end{enumerate}
\end{theo}
\begin{rem}\label{rem: numerical values}
\begin{itemize}\item Some effective bounds for $\beta_K$ are contained in Theorem 1.4  of \cite{LO77}. 
\item Numerical values for $a_K$, $b_K$ and $c_K$ are given without proof in \cite{Oe79}. The values  $a_K=\frac{4781}{96}\log(d_K)+\frac{58681}{113}n_K$, $b_K=\frac{23}{3}\log(d_K)+\frac{68}{3}n_K$ and $c_K=\frac{863}{31}n_K$ can be rigorously obtained from Theorem~$8.1$ of~\cite{Wi14}.\end{itemize}\end{rem} 

 In order to use  Theorem \ref{effectif}, we have to study 
 the contribution of powers of prime ideals. Using the Chebyshev estimation (\ref{Chebychev estimate}), we get,
 for any $X\geq2$,
\begin{align*}
\sum_{\substack{N(\mathfrak{p}^k)\leq X\\ k\geq 2\text{ or }N(\mathfrak{p})\text{ not prime}}}
\log N(\mathfrak{p})
&\leq e d_K\sum_{k\geq 2}\sum_{p\leq X^{\frac{1}{k}}}\log p\\
&\leq ed_K\left(X^{\frac{1}{2}}+\frac{\log X}{\log 2}X^{1/3}\right).
\end{align*}

Consequently, we have, for  $t\geq\exp\left(4d_K(\log \Disc(K))^2\right)$,
\begin{equation*}
 R(t)\ll
              d_Kt^{\frac{1}{2}}+\frac{t^{\beta(K)}}{\beta(K)}+t\exp\left(-c_1(d_K)^{-1/2}
              \left(\log t\right)^{1/2}\right)
\end{equation*}

By a straightforward calculation of primitive, we deduce 
from these estimations that we have, for $X_2\geq X_1\geq\exp\left(4d_K(\log \Disc(K))^2\right)$, 
\begin{align*}
 \left|\sum_{X_1<p\leq X_2}\frac{\log p}{p}\left(1-n_p(K)\right)\right| 
&\leq \frac{|R(X_1)|}{X_1}+\frac{|R(X_2)|}{X_2}
+\int_{X_1}^{X_2}\frac{|R(t)|}{t^2}\text{d}t\\
&\ll
   d_KX_1^{-\frac{1}{2}}+\frac{X_1^{\beta(K)-1}}{\beta(K)}+\exp\left(-c_1(d_K)^{-1/2}
              \left(\log X_1\right)^{1/2}\right),
\end{align*}
which implies the convergence of $\alpha(f)$.

In order to get a good estimation of the convergence speed, we now assume that the Riemann Hypothesis holds for $\zeta_{\mathbb{Q}}$ and $\zeta_K$.
It follows from Theorem \ref{effectif} that we have, for $t\geq2$,
\begin{align*}|R(t)|\leq ed_Kt^{1/2}+ed_Kt^{1/3}\log t/\log2+a_Kt^{1/2}+b_Kt^{1/2}(\log t)+c_Kt^{1/2}(\log t)^2.
\end{align*} 

As a consequence of the previous discussion, we can get that, for $X\geq2$, \begin{align*}
 \left|\sum_{X<p}\frac{\log p}{p}\left(1-n_p(K)\right)\right| 
&\leq \frac{|R(X)|}{X}
+\int_{X}^{+\infty
}\frac{|R(t)|}{t^2}\text{d}t\\
&\leq X^{-1/2}\Bigg(\left(3a_K+3ed_K+4b_K+16c_K+(3b_K+8c_K)\log X+3c_K(\log X)^2\right)\\
&+X^{-1/6}\frac{ed_K}{\log 4}\left(\frac{9}{2}+5\log X\right)\Bigg).
\end{align*}

 It follows that 
 the speed of convergence is given, for $X\geq\max(p_0,n_K)$, by
\begin{align}\nonumber\left|\alpha(f)-\sum_{p\leq X}\alpha_p(f)\right|\leq &X^{-1/2}\Big(3
\frac{eX^{1/2}}{X-1}+X^{-1/6}
\frac{ed_K}{\log 4}\left(\frac{9}{2}+5\log X\right)\\&+\left(3a_K+3ed_K+4b_K+16c_K+(3b_K+8c_K)\log X+3c_K(\log X)^2\right)\Big).\label{estimation convergence effective}
\end{align}
\begin{ex}\label{example} Using the best numerical values in Remark~\ref{rem: numerical values} we can certify effective bounds on $\alpha(F)$ for given binary forms $F$. Consider for example $F(X_1,X_2)=X_1^2+qX_2^2$ with $q=10^{30}+57$.
By computing the partial sum of $\alpha(F)$ for primes less than $X=40096176099$ we obtain:
\begin{equation*}
\left|\alpha(F)-2.39\right|<1.
\end{equation*}
This emphasizes the importance of obtaining small effective constants in Theorem~\ref{effectif}. 
\end{ex}

\section{Towards the average of $\alpha$ on a set of polynomials}\label{sec:average}
The polynomial selection stage of NFS consists in enumerating polynomials
$f=\sum_{i=0}^df_i x^i$ of a given degree and with a bound on each
coefficient $f_i$ and in selecting those with the best value of $\alpha$.
Some variants restrict the enumeration to a subset and a short list of
polynomials with a good $\alpha$ can be further tested with longer tests or
by direct sieving. In any case, by computing the average of $\alpha$ we
guarantee a value of $\alpha$ for the best polynomials.

During the polynomial selection in NFS, it is common to restrict the search
to a set of polynomials $f$ given by $\deg f$ and the size of each coefficient. For each pair $(m,d)$ of integers and each $d$-tuple $\mathbf{I}=I_0\times
\cdots\times I_{d-1}$
of intervals such that, for all $i$, $ I_i\subset[-m,m]$, we put
\begin{equation}
	E^{(1)}(m,d,\textbf{I})=\left\{f=x^d+\sum_{i=0}^{d-1} f_i x^i:
	(f_0,f_1,\ldots,f_{d-1})\in \textbf{I},\Disc(f)\neq 0 \right\}.
\end{equation}
Due to technical reasons, we now  now study the 
 average of $\alpha(f)$ on $E^{(1)}(m,d,\textbf{I})$ rather than $E(d,\textbf{I})$ defined by (\ref{definition EDI}). 

\begin{theo}
	For any given prime $p$, uniformly with respect to \textbf{I}, one has
	\begin{equation}
		\lim_{\substack{m\rightarrow\infty\\ \min_j|I_j|/d(\log d+\log m)\rightarrow\infty}}\frac{1}{\#
	E^{(1)}(m,d,\textbf{I})} \sum_{f\in E^{(1)}(m,d,\textbf{I})}
	\alpha_p(f)=\alpha_p(X).
	\end{equation}
\end{theo}

\begin{proof}
	In view of Proposition~\ref{prop:alpha_p} and  Theorem~\ref{alpha degre 1}, we can suppose that $d\geq2$ and write, for any prime $p$,	
	\begin{align*}\alpha_p(f)-\alpha_p(x)=\frac{p\log
	p}{p+1}\sum_{k\geq1}\frac{1-n_{p^k}(f)}{p^k}.\end{align*} 
	
	 For any pair $k$, we put 
	\begin{equation*}
		S_p(k,m,d,\textbf{I})=
		\sum_{\substack{f\in E(d,m,\textbf{I})}}(1-n_{p^k}(f)).
	\end{equation*}	
	Then we have
	\begin{align*}
	 \sum_{f\in
	E^{(1)}(m,d,\textbf{I})}(\alpha_p(f)-\alpha_p(x))=\Sigma_p^{(1)}(m,d,\textbf{I})+\Sigma_p^{(2)}(m,d,\textbf{I})
	,
	\end{align*}
	where 
	\begin{eqnarray*}
		\Sigma_p^{(1)}(m,d,\textbf{I})&=&\frac{p\log p}{p+1}\sum_{k\leq
	k_0}\frac{S_p(k,m,d,\textbf{I})}{p^k},\\
		\Sigma_p^{(2)}(m,d,\textbf{I})&=&\frac{p\log p}{p+1}\sum_{k\geq k_0}  \frac{S_p(k,m,d,\textbf{I})}{p^k}.
	\end{eqnarray*}

Using the definition of the discriminant, for any $f$ in $E^{(1)}(m,d,\textbf{I})$, we have the upper bound
\begin{align*}
|\Disc(f)|\leq (2d-1)! m^{2d-1}.
\end{align*}

Consider $k_0(p)= \left\lceil \log_p \left((2d-1)!m^{2d-1}\right)\right\rceil+\lceil \log_p(md)\rceil$.

	\textbf{Case 
	 $k\leq k_0(p)$.} 
	Since the elements of $E^{(1)}(m,d,\textbf{I})$ are monic, we have 
	\begin{align*}	
	\#\left\{f\in E^{(1)}(m,d,\textbf{I}), p^{jd}f\left(p^{-j}\right)\equiv0\pmod{p^k}
	\right\}=0
	.\end{align*}  
Consequently, we can write
	\begin{align*}
\Sigma_p^{(1)}(m,d,\textbf{I})=\frac{p\log p}{p+1}\sum_{
k\leq
	k_0}\frac{1}{p^k}\left(\#E^{(1)}(m,d,\textbf{I})-
	\sum_{r=0}^{p^k-1}
	\#\left\{f\in E^{(1)}(m,d,\textbf{I}), f(r)\equiv0\pmod{p^k}
	\right\}\right)
	.\end{align*}
	
	We consider first the cardinality of  $E^{(1)}(m,d,\textbf{I})$. Given
	$(f_1,\ldots,f_{d-1})\in I_1\times\cdots I_{d-1}$, the polynomial $dx^{d-1}+\sum_{i=1}^{d-1} if_ix^{i-1}$ 
	has at most $d-1$ complex
	roots. For each such root $z$, there is exactly one value of $f_0\in I_0$
	such that $\sum_{i=0}^d f_i z^i=0$. Hence there are at most $d|\mathbf{I}|/|I_0|$ polynomials $f$ of zero discriminant and coefficients in $\textbf{I}$. It follows that
	\begin{align*}
	E^{(1)}(m,d,\textbf{I})&=\#\left\{(f_0,\dots,f_{d-1})\in\textbf{I}\right\}
	-\#\left\{(f_0,\dots,f_{d-1})\in\textbf{I}:
	\Disc(x^d+\sum_{i=0}^{d-1}f_ix^i)=0\right\}
	\\&=\left|\textbf{I}\right|\left(1+O\left(\frac{d}{\min_j\left|I_j\right|}\right)\right)
	.\end{align*}
	
	Let $k\leq k_0(p)$ be an integer and $r\in[0,p^k-1]$. For each $(d-1)$-tuple
	$(f_1,\ldots,f_{d-1})\in I_1\times\cdots\times I_{d-1}$, the number of
	values $f_0$ such that $f(r)\equiv0\pmod{p^k}$ is $\left\lfloor
	\frac{|I_0|}{p^k}\right\rfloor+\epsilon$ with $\epsilon=0$ or $1$. Hence,
	it follows that
	\begin{align*}&\#\left\{f\in E^{(1)}(m,d,\textbf{I}), f(r)\equiv0\pmod{p^k}
	\right\}\\=&
	\#\left\{(f_0,\ldots,f_{d-1})\in\textbf{I}: f(r)\equiv0\pmod{p^k}\right\}\\
	&+O\left(\#\left\{(f_0,\ldots,f_{d-1})\in\textbf{I}: \Disc(x^d+\sum_{i=0}^{d-1}f_ix^i)=0\right\}\right)\\
		\\=&
		\left(\frac{|I_0|}{p^k}+O(1)\right)\frac{|\textbf{I}|}{|I_0|}\left(1+O\left(\frac{d}{\min_j|I_j|}\right)\right)+O\left(\frac{d|\textbf{I}|}{\min_j\left|I_j\right|}\right)\\
		=&\frac{|\textbf{I}|}{p^k}+O\left(\frac{|\textbf{I}|d}{\min_j|I_j|}\right).\end{align*}  

	It results that
	\begin{align*}
	\Sigma_p^{(1)}(m,d,\textbf{I})&\ll 
	\log p\sum_{k\leq k_0(p)}\frac{|\textbf{I}|d}{\min_j|I_j|}
	\\
	&\ll k_0(p)\log p \frac{|\textbf{I}|d}{\min_j|I_j|}.
	\end{align*}

	\textbf{Case 
	$k\geq k_0(p)$.} 
	Due to the choice of $k_0(p)$, we have $k_0\geq 2\val_p \Disc(f)$ for all polynomials $f$ in $E(m,d,\textbf{I})$. By Lemma~\ref{Nagel}, for all $k\geq k_0(p)$, we have 
	\begin{align*}n_{p^k}(f)\ll d\Disc(f)^2,
\end{align*}
which is further upper bounded by $ (2d-1)!dm^{2d-1}\leq p^{k_0(p)}/m$. We deduce that
	\begin{align*}
	\Sigma_p^{(2)}(m,d,\textbf{I})
	&\ll
		d((2d-1)!m^{2d-1})^2\log p\sum_{\substack{
		k>k_0(p)}}p^{-k}\left|
		E^{(1)}(m,d,\textbf{I})
		\right|\\
	&\ll d((2d-1)!m^{2d-1})^2|\textbf{I}|
	\frac{\log p}{p^{k_0(p)}}\\
	&\ll \frac{|\textbf{I}|}{m}. 
	\end{align*}
	
	When combining the bounds on $\sum^{(1)}_p(m,d,\textbf{I})$ and $\sum^{(2)}_p(m,d,\textbf{I})$, we obtain that, uniformly for $p\geq1$, we have
	\begin{align}\label{estimation petits p}
	\sum_{f\in
	E^{(1)}(m,d,\textbf{I})}(\alpha_p(f)-\alpha_p(x))\ll|\textbf{I}|\left(\frac{1}{m}+\frac{d(\log d+\log m)}{\min_j|I_j|}\right).
	\end{align}
\end{proof}

In view of the previous theorem, it seems to be interesting to compute the value  of $\alpha(X)$. This is the aim of the following proposition.

\begin{prop}\label{alpha degre 1}
	Let $g=aX+b\in\Z[X]$ be a polynomial with $\gcd(a,b)=1$. Then we have
\begin{align*}
 \alpha(g)=12\log A-\gamma-\log(2\pi)\approx 0.56 
.\end{align*}
where $A$ denotes the Glaisher-Kinkelin constant and $\gamma$ denotes the Euler-Mascheroni constant.
\end{prop}
\begin{proof}
	Since $f$ has degree $1$ and $\gcd(a,b)=1$, we have, for every prime $p$ and $k\geq1$, 
	\begin{align*} 
	 n_{p^k}(f)=1.
	 \end{align*} Consequently, it follows from Proposition \ref{prop:alpha_p} that
	 \begin{align*}
	\alpha(f)=\sum_{p}\frac{\log p}{p-1}\left(1-\frac{p}{p+1}\right)=
 \sum_{p}\frac{\log p}{p^2-1}.
\end{align*}
From the formula
\begin{align*}
 \frac{\zeta_{\mathbb{Q}}'(s)}{\zeta_{\mathbb{Q}}(s)}=-\sum_{p}\frac{\log p}{p^s-1},
\end{align*}
which holds for any complex $s$ such that $\Re(s)>1$, we deduce that
\begin{align*}
 \alpha(f)=\sum_{p}\frac{\log p}{p^2-1}=-\frac{\zeta_{\mathbb{Q}}'(2)}{\zeta_{\mathbb{Q}}(2)}.
\end{align*}
The result is then a direct consequence of the formulas
\begin{align*}
\zeta_{\mathbb{Q}}(2)=\frac{\pi^2}{6}\qquad\text{ and }\qquad\zeta_{\mathbb{Q}}'(2)=\frac{\pi^2}{6}\left(\gamma+\log(2\pi)-12\log A\right).
\end{align*}
\end{proof}

We can remark that this proposition asserts that $\alpha(g)=\alpha(X)$ for any linear polynomial $g$. This observation is a new argument towards the direction that the polynomial selection is essentially not influenced by the linear polynomial. 

\section{A theoretical modification of NFS}\label{sec:modification}
\subsection{The algorithm}
The main goal of this section is to prove smoothness results for binary forms of degree $2$. This case can be treated with multiplicative methods since the values of a quadratic binary form are norms of arbitrary integer elements of a quadratic field. The same theorems apply to binary forms of higher degrees if we modify the algorithm as below. By doing so, we transfer the difficulty from the field of analytic number theory to that of algorithmic number theory. 

In short, in our modification of NFS, instead of considering elements $a-b\theta$ of $\Q(\theta)$, we consider arbitrary elements $a_0+a_1\omega+\cdots a_{d-1} \omega^{d-1}$ of norm bounded by a constant, where $d$ is the degree of the defining polynomial $f$. In more detail, the new version of the algorithm is as follows. We select
two polynomials $f$ and $g$, with $g$ linear such that
there exists an integer $m$ such that $f(m)\equiv g(m)\equiv 0\mod N$. We
use the same factor base as in the classical version of NFS, i.e. if $B$ is the
smoothness bound, the factor base includes degree-$1$ ideals in the number
field of $f$ and primes up to $B$. Let $\omega$ be a root of $f$ in its
number field. We set $X_f$ and $X_g$ to the maximal value of
$\Norm(a_0+a_1\omega)$  and $|a_0+a_1m|$ respectively when $a_0$ and $b_0$ are bounded by the constant used in NFS. Next we collect primitive polynomials
$P(x)=a_0+a_1x+\cdots+a_{d-1}x^{d-1}$ such that 
\begin{itemize}
\item $(a_0,\dots,a_{d-1})=1$
\item $|\Norm(P(\omega))|\leq X_f\text{ and }|P(m)|\leq X_g.$
\item $\Norm(P(\omega))$ and
$|P(m)|$ are $B$-smooth.
\end{itemize}
 Each polynomial $P$ allows us to obtain a relation as explained by Joux, Lercier, Smart and Vercauteren in~\cite{JLSV06}. Finally, we use the linear system to obtain a non-trivial solution of equation $X^2\equiv Y^2\pmod N$ by following step by step the classical variant of NFS.  

The practicality of this modification will be investigated by the first author in a future work. The main difficulty is to enumerate the ideals whose norm is bounded by a given constant. 

\subsection{The smoothness probability : general case}

Let $\omega$ be an algebraic integer, non rational, and $K=\Q(\omega)$. In view of the previous discussion, we now focus on the study of the cardinality of
\begin{align*}
&\Bigg\{(a_0,\dots,a_{d-1})\in\Z^d: \gcd(a_0,\dots,a_{d-1})=1, N(a_0+\dots+a_{d-1}\omega^{d-1})\leq x\qquad\qquad\qquad\qquad\\&\pushright{\text{ and }
P\left(N(a_0+\dots+a_{d-1}\omega^{d-1})\right)\leq B\Bigg\}}.
\end{align*}

If the unit group $U_K$ is infinite (this is the case when $d_K\geq3$ or $K$ is a real quadratic field), such a set is infinite. 
However, we can remark that the ideals $\mathfrak{I}$ generated by its elements are primitive, namely that, for any prime $p$, $p\mathcal{O}_K\nmid \mathfrak{I}$.
Consequently, it makes sense to concentrate ourself to  the cardinality 
\begin{align*}
\Psi_K^{(1)}(x,B):=\#\left\{\mathfrak{I}\text{ primitive}: N(\mathfrak{I})\leq x\text{ and }
P(N(\mathfrak{I}))\leq B
\right\}.
\end{align*} 

A standard way -- the one followed here --   to get an asymptotic formula for $\Psi_K^{(1)}(x,B)$ consists to apply to the Dirichlet series $\mathcal{F}_K(s)$ defined by
\begin{align*}
\mathcal{F}_K(s):=\sum_{\substack{\mathfrak{I}\text{ primitive}}}\frac{1}{N(\mathfrak{I})^s}
\end{align*}
some results of complex analysis, such as Perron's formula.
It is consistent to take a look at the shape of $\mathcal{F}_K(s).$  Using the inclusion–exclusion principle, we first remark that we have, for $\Re(s)>1$,
\begin{align}
 \mathcal{F}_K
 &=\sum_{m\geq 1}\mu(m)\sum_{\substack{m\mathcal{O}_K|\mathfrak{I}}}\frac{1}{N(\mathfrak{I})^s}
 =\zeta_K(s)\zeta_{\mathbb{Q}}(d_Ks)^{-1}.\label{inclusion-exclusion F}
\end{align}
Moreover, using the properties of the Riemann zeta function, it is immediate that $\zeta_{\mathbb{Q}}(d_Ks)^{-1}$ is absolutely convergent for $\Re(s)>\frac{1}{d_K}$. 

In view of the previous discussion, we are now in capacity to use asymptotic results of Hanrot, Tenenbaum and Wu \cite{HTW08}. We obtain the following theorem. 
\begin{theo}\label{corollaire smooth}
 Let $K$ be a number field of degree $d_K\geq2$
 . 
Then, there exists $C>0$ such that, for any $J\geq0$ and $\varepsilon>0$, we have, 
 uniformly for $\exp\left((\log\log x)^{5/3+\varepsilon}\right)\leq B\leq x$ and 
 \begin{align*}
  0<u<J+1\Rightarrow \left\{
  u\right\}
  >C(J+1)\frac{\log_2 B}{\log B} 
\end{align*}
  \begin{equation}\label{cardinal K}
 \Psi_K^{(1)}(x,B) =
  x\left(\sum_{j=0}^J\gamma_j(K)\frac{\rho^{(j)}(u)}{(\log B)^j}+
O\left(\rho(u)\left(\frac{\log(u+1)}{\log B}\right)^{J+1}
\right)\right),
\end{equation}
where 
\begin{align*}
 \gamma_j(K)\sum_{j_1+j_2=j}\frac{1}{j_1!j_2!}
  \left.
\frac{\partial^{j_1}(1-s^{-1})\zeta_K(s)}{\partial s^{j_1}}\right|_{s=1}
\left.\frac{\partial^{j_2}\zeta_{\mathbb{Q}}(d_Ks)^{-1}}{\partial s^{j_2}}\right|_{s=1}.
\end{align*}

In particular, we have \begin{align*}
\gamma_0(K)=\frac{\lambda_K}{\zeta_{\mathbb{Q}}(d_K)}
 \end{align*}
  and 
 \begin{align}\nonumber
\gamma_1(K)=&\gamma_0(K)\left(
\gamma-1+
\sum_p\log p\left(
 \frac{1}{p-1}-\cont_p(K)\right)\right),
\end{align}
with
\begin{align}
\cont_p(K)&=\left(\sum_{k\geq1}\frac{k\#|\left\{\mathfrak{I}\text{ primitive},N(\mathfrak{I})=p^k\right\}}{p^k}
 \right) \left(\sum_{k\geq0}\frac{\#\left\{\mathfrak{I}\text{ primitive},N(\mathfrak{I})=p^k\right\}}{p^k}
 \right)^{-1}.\nonumber
 \end{align}
\end{theo}
\begin{proof}
 In view of Equation~\eqref{inclusion-exclusion F},  it is immediate that $\mathcal{F}_K(s)$ satisfies the 
 Condition~$(1.7)$ of \cite{HTW08}. Moreover, as it is noted in Section~$2.3$ of \cite{HTW08},
 Theorem~II$.1.13$ of \cite{Te95a} implies that, for any $\frac{1}{d(K)}<\delta<1$ and 
 uniformly for $\Re(s)\geq\delta$, we have
 \begin{align}\label{somme tronquee y}
  \sum_{P(n)\leq B}\frac{\mu(n)}{n^{d_Ks}}=\sum_{n}\frac{\mu(n)}{n^{d_Ks}}+O\left(\frac{1}{B^{1-\delta}}\right).
 \end{align}
Consequently, we can apply successively 
Theorem 1.2 and Theorem 1.1 of \cite{HTW08} to deduce (\ref{cardinal K}).

The statement on the values $\gamma_0(K)$ and $\gamma_1(K)$ follows from the fact that
\begin{align*}
 \left.\left(\frac{
 \frac{\partial \mathcal{F}_K(s)}{\partial s}}{\mathcal{F}_K(s)} - \frac{
\frac{\zeta_{\mathbb{Q}}(s)}{\partial s}}{\zeta_{\mathbb{Q}}(s)}\right)\right|_{s=1}=&
\sum_p\log p\left(
 \frac{1}{p-1}-\text{cont}_p(K) \right).
\end{align*}
\end{proof}

\subsection{The smoothness probability : imaginary quadratic case}\label{The smoothness probability : imaginary quadratic case}

Let $f$ be an irreducible quadratic polynomial. Its discriminant $\Disc(f)$ is  a fundamental discriminant if it satisfies one of the following conditions :
\begin{itemize}\item $\Disc(f) \equiv 1 \pmod{4}$ and is square-free,\item    $\Disc(f) = 4m$ where $m \equiv 2\text{ or }3 \pmod{4}$ and $m$ is square-free.\end{itemize}
We now apply the previous result to get an asymptotic estimation related to the proportion of smooth values of quadratic binary forms with fundamental negative discriminant.

\begin{theo}\label{theoreme principal quadratique}
 Let $F(X_1,X_2)\in\Z[X_1,X_2]$ be a primitive and irreducible quadratic form such that $\Disc(F)$ is negative and fundamental. Let $\mathcal{K}_F$ the compact defined by
 \begin{displaymath}
 \mathcal{K}_F:=\left\{(x_1,x_2)\in\R^2:|F(x_1,x_2)|\leq 1\right\}
 \end{displaymath} Then, there exists $\kappa>0$ such that, for any  $\varepsilon>0$, we have, 
 uniformly for $\exp\left((\log\log x)^{5/3+\varepsilon}\right)\leq B\leq x(\log x)^{-\kappa}$,
  \begin{align}
  \frac{\Psi^{(1)}_F(\mathcal{K}_F,x,B)
  }{\Psi^{(1)}_F(\mathcal{K}_F,x,x)
  }=
  \frac{\Psi(xe^{\alpha(f)},B)
  }{
  xe^{\alpha(f)}
  }
  \left(1+O\left(\frac{(\log(u+1))^2}{(\log B)^2}\right)\right).
  \label{cardinal K degre 2}
\end{align}

\end{theo}
\begin{proof}
Let $\omega$ be a root of $f(X)=F(X,1)$ and $K:=\Q(\omega)$. Since $\Disc(f)$ is a fundamental discriminant, we have $\Disc(K)=\Disc(f)$. Moreover, there exists a basis $(\omega_1,\omega_2)$ of $\mathcal{O}_K$ such that, for any integers $a$ and $b$, one has
\begin{equation*}
F(a,b)=\Norm(a\omega_1+b\omega_2).
\end{equation*}

Since $U_K$ is finite,  we have
\begin{align*}
\Psi^{(1)}_F(\mathcal{K}_F,x,B)
=&\#\left\{\omega=(a\omega_1+b\omega_2)\in \OO_K~:~(a,b)=1, \Norm(a\omega_1+b\omega_2)|\leq x, P(\Norm((a\omega_1+b\omega_2))\leq B
\right\}\\
= &\left|U(K)\right|
\#\left\{\mathfrak{I}\text{ principal ideal}~:~\mathfrak{I}\text{ is primitive}, N(\mathfrak{I})\leq x, P(N(\mathfrak{I}))\leq B
\right\}
\end{align*}
In order to pick up ideals from the class $\text{Cl}(\mathcal{O}_K)$, i.e. principal ideals, we can consider the group $\widehat{G_K}$ of the multiplicative characters of the class group $G_K$. By the orthogonality property of characters, we have 
\begin{align*}
\#\left\{\mathfrak{I}\in \text{Cl}(\mathcal{O}_K):N(\mathfrak{I})\leq x,\mathfrak{I}\text{ primitive},P(N(\mathfrak{I})))\leq B
\right\}
=\frac{1}{|G_K|}\sum_{\chi\in\widehat{G_K}}\Psi^{(1)}(x,B;\chi),
\end{align*}
where 
\begin{align*}
\Psi^{(1)}(x,B;\chi)=\sum_{
\substack{\mathfrak{I}\text{ primitive}\\ N(\mathfrak{I})\leq x\\P(N(\mathfrak{I}))\leq B
}}\chi(\mathfrak{I}).
\end{align*}
\textbf{Contribution of nontrivial characters:}

Since $\text{Cl}(p\mathcal{O}_K)$ is the identity element of the class group $G_K$, the inclusion-exclusion principle implies that
\begin{align*}
\sum_{\substack{\mathfrak{I}\text{ primitive}}}\frac{\chi(J)}{N(J)^s}
=
\sum_{\substack{\mathfrak{I}}}\frac{\chi(\mathfrak{I})}{N(\mathfrak{I})^s}
\left(\prod_p\left(1-\frac{1
}{p^{2s}}\right)\right)^{-1}.
\end{align*}
Consequently, we can adapt, step by step, the proof of Theorem \ref{corollaire smooth} to deduce that, for any $\varepsilon$ and uniformly for
\begin{align*}x\geq 3\quad\text{ and }\quad
\exp\left((\log\log x)^{5/3+\varepsilon}\right)\leq B\leq x,
\end{align*}
we have
\begin{align*}
\sum_{\substack{\mathfrak{I}\text{ primitive}}}\frac{\chi(\mathfrak{I})}{N(\mathfrak{I})^s}\ll
x\rho(u)\exp\left(-(\log B)^{3/5-\varepsilon}\right).
\end{align*}

This procedure is essentially made in \cite{Te90a} and \cite{FT91}.

\noindent\textbf{Contribution from the trivial character :}

For the principal character, denoted by $\chi_0$, we use Theorem \ref{corollaire smooth}. There exists $c>0$ such that, for any  $\varepsilon>0$, we have, 
 uniformly for \begin{displaymath}x\geq 3\text{ and }\exp\left((\log\log x)^{5/3+\varepsilon}\right)\leq B\leq x(\log x)^{-c},
 \end{displaymath} 
\begin{align*}
\Psi^{(1)}(x,B;\chi_0)=
  x\left(\gamma_0(K)\rho(u)+\gamma_1(K)\frac{\rho'(u)}{\log B}+
O\left(\rho(u)\left(\frac{\log(u+1)}{\log B}\right)^{2}
\right)\right),
\end{align*}
where $
\gamma_0(K)=\frac{6\lambda_{K}}{\pi^2}$
 and  
\begin{align*}
\gamma_1(K)=&\gamma_0(K)\left(
\gamma-1+
\sum_p\log p\left(
 \frac{1}{p-1}-\text{cont}_p(K)\right)\right).
\end{align*}

Using the decomposition of rational primes into ideals of $\mathcal{O}_K$ (see for example the discussion in Section~$6.4$ of \cite{Bu89}), we can note that
\begin{align*}
\#\left\{\mathfrak{I}\text{ primitive},N(\mathfrak{I})=p^k\right\}=
\left\{\begin{array}{ll}
0&\text{ if }p\mid\Disc(K)\text{ and }k\geq2,\\
n_p(K)&\text{ if }k=1\text{ or }p\nmid\Disc(K),
\end{array}\right.
\end{align*}and therefore
\begin{displaymath}
\cont_p(K)=\left\{\begin{array}{ll}
\frac{1}{p+1}&\text{ if }p|\Disc(K),\\
\frac{p}{p+1}\frac{n_p(K)}{p-1}&\text{ otherwise}.
\end{array}\right.\end{displaymath}
A careful study of $\cont_p(f)$ implies that we have actually
\begin{align}
\text{cont}_p(K)
=\text{cont}_p(f)\label{contK contp}
\end{align}
To see this, assume first that $p|\Disc(K)$. In view of the hypothesis on $\Disc(K)$, a straightforward computation implies that $n_{p}(f)=1$ and $n_{p^k}(f)=0$ for $k\geq 2$, and therefore Equation~(\ref{contK contp}) holds. We consider now primes $p$ which do not divide $\Disc(K)$, for which we must show that $n_p(f)=n_K(f)$ (Hensel's Lemma allows to obtain $n_{p^k}(f)=n_{p^k}(K)$ for $k\geq 2$). If $p$ does not divide $2F(1,0)F(0,1)$, since the index is $1$ or $2$, Dedekind's result states that $n_p(f)=n_p(K)$. If $p$ is an odd prime which divide $F(1,0)F(0,1)$, it is not difficult, using the decomposition of $p$ in $\mathcal{O}_K$, to see that $n_p(f)=n_p(K)=2$. If $p=2$ and (at least) one of $F(1,0)$ and $F(0,1)$ is even, then $\Disc(K)\equiv1\pmod{8}$, which implies that $n_2(K)=2$. But then $F(0,1)$ and $F(1,1)$ are even and one obtains $n_2(f)=2=n_2(K)$. Finally, if $p=2$ does not divide $F(0,1)$ nor $F(1,0)$, all the coefficients of $F$ are odd and then $n_2(f)=0$. Since, in this case, $\Disc(K)\equiv5\pmod{8}$, we have also $n_2(K)=0=n_2(f)$.
For the remaining primes, we have by Lemma \ref{Nagel} that $n_{p^k}(f)=n_p(K)$ for any $k\geq1$ which implies (\ref{contK contp}).  
 
 From this discussion, it finally follows that
 \begin{align*}
\Psi^{(1)}_F(\mathcal{K}_F,x,B)
  =&\frac{6\lambda_{K}}{\pi^2|G_K|}x\left(\rho(u)+(\gamma-1+\alpha(f) )\frac{\rho'(u)}{\log B}+O\left(\rho(u)\frac{(\log (u+1))^2}{(\log B)^2}\right)\right).
 \end{align*}
 
Using the standard Selberg-Delange's method instead of Theorem \ref{corollaire smooth} (see  \cite{Te95a}), we can also prove that, for any $\varepsilon>0$,   we have
 \begin{align*}
 \Psi^{(1)}_F(\mathcal{K}_F,x,x):=\#\left\{(a,b)\in\Z^2:(a,b)=1, |F(a,b|\leq x \right\}
 =\frac{6}{\pi^2|G_K|}x+O\left(x\exp\left(-\log x)^{3/5-\varepsilon}\right)\right).
 \end{align*}

From Theorem \ref{developpement asymptotique Saias}, we see also that  for any  $\varepsilon>0$ and
 uniformly for
 \begin{displaymath} x\geq3\text{ and }\exp\left((\log\log x)^{5/3+\varepsilon}\right)\leq B\leq x(\log x)^{-c},\end{displaymath} we have
 \begin{align*}
 \Psi(x,B)=x\left(\rho(u)+(\gamma-1)\frac{\rho'(u)}{\log B}+O\left(\frac{(\log(u+1))^2}{(\log B)^2}\right)\right).
 \end{align*} 
This enables us to estimate the right-hand term of Equation~\ref{cardinal K degre 2} and to deduce the result.
\end{proof}

\begin{rem}
The theorem above encompasses a large set of binary forms. For example, since the quadratic binary form $F=X_1^2+qX_2^2$ defined in Example~\ref{example} has fundamental discriminant and $\alpha(F)$ is positive, we know that asymptotically it has less smooth values than the random integers of same size. Nevertheless, many examples of binary forms $F'$ with good values of $\alpha(F')$ have non fundamental disciminants.
\end{rem}
\section{Conclusion and open questions}\label{sec:conclusion}
The results in this article establish a rigorous connection between Murphy's $\alpha$ and a polynomial's efficiency in NFS. On can improve the speed of the algorithm by studying $\alpha$ and, in particular, the following questions:
\begin{itemize}
\item What is the maximum value of $\alpha$ on a given set $E(d,\textbf{I})$? Indeed, if a polynomial with a good value of $\alpha$ is found, one can end the polynomial selection phase, reducing therefore the time spent in this phase of the algorithm.
\item Can one define a variance of $\alpha$? Indeed, experiments indicate that, uniformly on the ideals products $\textbf{I}$, the distribution of the values of $\alpha$ on a set $E(d,m,\textbf{I})$ converges to a Gaussian distribution when $m$ tends to infinity. If one can define and compute the variance of $\alpha$, one will be able to find a good trade-off between the time spent to select a good polynomial and the time used to collect relations using that polynomial.

\end{itemize} 

\bibliographystyle{alpha}
\bibliography{biblio}
\end{document}